\documentclass[journal]{IEEEtran}
\ifCLASSINFOpdf
\else
\fi
\ifCLASSOPTIONcompsoc
 \usepackage[caption=false,font=normalsize,labelfont=sf,textfont=sf]{subfig}
\else
 \usepackage[caption=false,font=footnotesize]{subfig}
\fi

\usepackage{graphicx}
\usepackage{amsmath}
\usepackage{mathtools}
\usepackage{epsfig}
\usepackage{float}
\usepackage{lipsum}
\usepackage{stfloats}
\usepackage{array}
\usepackage{amssymb}
\usepackage{cite}
\usepackage{url}
\usepackage{algorithm}
\usepackage{algpseudocode}
\usepackage{multirow}
\usepackage{fancyh}
\usepackage{upgreek}
\usepackage{dsfont}
\usepackage{gensymb}
\usepackage{physics}
\usepackage{amsthm}

\usepackage{multicol, blindtext}
\usepackage{mathtools, cuted}

\usepackage{color}
\usepackage{bm}
\usepackage{booktabs}
\usepackage{tablefootnote}
\usepackage{threeparttable}
\usepackage{makecell}

\usepackage{tikz}
\usepackage{pgfmath, pgfplots}
\usetikzlibrary{3d, shapes.geometric, calc}
\usepackage{color}

\usepackage{pifont}

\def\conv{{\rm{Conv}}}

\def\zero{{\boldsymbol{0}}}

\def\bb0{{\mathbb{0}}}

\def\bb{{\boldsymbol{b}}}

\def\b0{{\boldsymbol{0}}}

\def\b{{\mathrm{b}}}

\def\r0{{\mathbf{0}}}

\def\bbC{{\mathbb{C}}}

\def\bsf0{{\bm{\mathsf{0}}}}

\def\N0{{N_{\mathrm{0}}}}

\def\bsf{{\boldsymbol{s}_\mathrm{f}}}

\newcommand{\be}{\begin{equation}}
\newcommand{\ee}{\end{equation}}
\newcommand{\bal}{\begin{align}}
\newcommand{\eal}{\end{align}}

\theoremstyle{remark}
\newtheorem{theorem}{Theorem}
\newtheorem{lemma}{Lemma}

\normalsize
\begin{document}
%
\title{Optimum Discrete Beamforming via \\ Minkowski Sum of Polygons
\thanks{H. Do and A. Lozano are with Univ. Pompeu Fabra, 08018 Barcelona (e-mail:\{heedong.do, angel.lozano\}@upf.edu).
Their work is supported by ICREA and by the Maria de Maeztu Units of Excellence Programme CEX2021-001195-M funded by MICIU/AEI/10.13039/501100011033.}
}
\author{\IEEEauthorblockN{Heedong~Do},
 {\it Member,~IEEE},
 \and
 \IEEEauthorblockN{Angel~Lozano},
{\it Fellow,~IEEE}
\vspace{-4mm}
}
\maketitle



%


\maketitle

\begin{abstract}
This letter casts the problem of optimum discrete beamforming as the computation of the Minkowski sum of convex polygons, which is itself a convex polygon. The number of vertices of the latter is at most the sum of the number of vertices of the original polygons, enabling its efficient computation. This original and intuitive formulation confirms that the optimum beamforming solution can be found efficiently.  
\end{abstract}



%
\IEEEpeerreviewmaketitle




\section{Optimum Discrete Beamforming}


Consider a transmitter equipped with an $N$-antenna array and a single-antenna receiver.
The phase shift for the $n$th transmit antenna is chosen from a finite set $\Theta_n\subset \bbC$. Denoting its channel to the receiver by $h_n\in\bbC$, the maximization of the beamforming gain can be formulated as
\begin{align}
    \max_{w_1,\ldots,w_N} & \quad \bigg| \sum_n w_n h_n \bigg| \label{original_problem}\\
    \text{s.t.} & \quad w_n\in\Theta_n. \nonumber
\end{align}

As shown in \cite{mackenthun1994fast, sweldens2001fast, motedayenaval2003polynomial, alevizos2016log, deng2019mmwave, sanchez2021optimal, zhang2022configuring, ren2023ieee, vardakis2023intelligently, pekcan2024achieving, sanjay2024optimum, kutay2024received}, the solution to \eqref{original_problem} can be efficiently computed without exhaustively searching over $\Theta_1\times \cdots \times \Theta_N$. This letter reaffirms this finding,
with an alternative formulation that is particularly insightful. 
%
%
The key observation is that \eqref{original_problem} is equivalent to 
\begin{align}
    \max_{z} & \quad |z| \label{minkowski_sum_problem} \\
    \text{s.t.} & \quad z\in h_1\Theta_1 + \cdots + h_N \Theta_N, \nonumber
\end{align}
where the summations 
are Minkowski sums of sets. 


\begin{algorithm}[t]
\caption{Minkowski sum of polygons}\label{algo:compute_minkowski_sum}
\begin{algorithmic}
    \State \textbf{Inputs:} convex polygons as lists of vertices:
    \begin{align*}
    \begin{matrix}
       \text{polygon}& \hspace*{-0.8em} P_1 & p_{1,1} & p_{1,2} & \cdots & p_{1,M_1}\\
       \vdots \\
       \text{polygon}& \hspace*{-0.8em} P_N & p_{N,1} & p_{N,2} & \cdots & p_{N,M_N}
    \end{matrix}
    \end{align*}
    with the vertices in each row arranged counterclockwise from the one with smallest imaginary part, with the convention that $p_{n,M_n+1} = p_{n,1}$.
    \State \textbf{Outputs:} an array of indices:
    \begin{align*}
    \begin{matrix}
        \text{1st vertex}& m_{1,1} & m_{1,2} & \cdots & m_{1,N}\\
        & \vdots & \vdots & \ddots & \vdots\\
        \text{$K$th vertex} & m_{K,1} & m_{K,2} & \cdots & m_{K,N}
    \end{matrix}
    \end{align*}
    where $K = \sum_n M_n$.
    The $k$th vertex of $P_1  + \cdots + P_N$ is
    \begin{align*}
        p_{1,m_{k,1}} + p_{2,m_{k,2}} + \cdots + p_{N,m_{k,N}}
    \end{align*}
    \hrule
    \vspace*{0.3em}
    \State Initialize an empty list
    \For{$n=1,2,\ldots,N$}
        \For{$m=1,2,\ldots,M_n$}
            \State Compute $p_{n,m+1}-p_{n,m}$, the $m$th edge of the $n$th \\
            \hspace*{2.7em} polygon
            \State Add a pair $(\arg(p_{n,m+1}-p_{n,m}), n)$ to the list with\\
            \hspace*{2.7em} a convention that the argument is in $[0,2\pi)$
        \EndFor
    \EndFor
    \vspace*{0.1em}
    \State Sort the list by the first component of the items, from smallest to largest
    \State $(\ell_1,\ell_2,\cdots,\ell_N) \gets (1,1,\cdots,1)$
    \For{$k=1,2,\ldots,K$}
        \State $(m_{k,1}, m_{k,2}, \cdots, m_{k,N}) \gets (\ell_1,\ell_2,\cdots,\ell_N)$
        \State $\ell_n \gets \ell_n + 1$, where $n$ is the second component of the\\
        \hspace*{1.25em} $k$th item in the sorted list
    \EndFor
\end{algorithmic}
\end{algorithm}

\begin{theorem}
\label{theorem:evaluation_at_vertices}
The problem in \eqref{minkowski_sum_problem} can be reformulated as
\begin{align}
    \max_{z} & \quad |z| \\
    \text{s.t.} & \quad z \text{ is a vertex of } h_1\,\conv\,\Theta_1 + \cdots + h_N\, \conv\,\Theta_N , \nonumber
\end{align}
where $\conv(\cdot)$ returns the convex hull of a set.
\end{theorem}
\begin{proof}
See Appendix.
\end{proof}


The implications of Thm. \ref{theorem:evaluation_at_vertices}, fleshed out in the next section,
extend beyond classical beamforming because \eqref{original_problem} subsumes
\begin{align}
    \max_{w_1,\ldots,w_N} & \quad \bigg|h_0 + \sum_n w_n h_n \bigg| \label{ris_problem}\\
    \text{s.t.} & \quad w_n\in\Theta_n, \nonumber
\end{align}
which is relevant because the beamforming optimization for reconfigurable intelligent surfaces
\cite{hashemi2024optimal}
\begin{align}
    \max_{w_1,\ldots,w_N} & \quad \bigg|h_0 + \sum_n w_n h_n \bigg|\\
    \text{s.t.} & \quad w_n\in \{tw: t\in[0,1], w\in \Theta_n\} , \nonumber
\end{align}
can be mapped back to \eqref{ris_problem} via
$\Theta_n \gets \{\zero\} \cup \Theta_n$ given that
\begin{align}
    \conv \{tw: t\in[0,1], w\in \Theta_n\} = \conv\big(\{\zero\}\cup \Theta_n\big) .
\end{align}

\begin{figure*}
    \centering
    \subfloat[Convex polygon $P_1$]
    {
        \begin{tikzpicture}[>=stealth]
        \begin{scope}
            \clip (-2.1,-2.1) rectangle (2.1,2.1);
            \draw[line width=0.5pt, ->] (-2.1,0) -- (2.1,0) node (X) [right,xshift = -0.6 cm, yshift = -0.2 cm]{Re};
            \draw[line width=0.5pt, ->] (0,-2.1) -- (0,2.1) node[above,xshift = -0.3 cm, yshift = -0.4 cm]{Im};
            \coordinate (P1) at (0.4,-0.7);
            \coordinate (P2) at (0.6,0.7);
            \coordinate (P3) at (-0.4,0.6);
            \coordinate (P4) at (-0.3,-0.4);
            \filldraw[color=red] (P1) circle (1.5pt) node[below]{$p_{1,1}$};
            \filldraw[color=red] (P2) circle (1.5pt) node[above]{$p_{1,2}$};
            \filldraw[color=red] (P3) circle (1.5pt) node[above]{$p_{1,3}$};
            \filldraw[color=red] (P4) circle (1.5pt) node[below]{$p_{1,4}$};
            \filldraw[line width=1pt, color=red, fill=red, fill opacity=0.3] 
            (P1) -- (P2) -- (P3) -- (P4) -- cycle;
        \end{scope}
        \end{tikzpicture}
    }
    \subfloat[Convex polygon $P_2$]
    {
        \begin{tikzpicture}[>=stealth]
        \begin{scope}
            \clip (-2.1,-2.1) rectangle (2.1,2.1);
            \draw[line width=0.5pt, ->] (-2.1,0) -- (2.1,0) node (X) [right,xshift = -0.6 cm, yshift = -0.2 cm]{Re};
            \draw[line width=0.5pt, ->] (0,-2.1) -- (0,2.1) node[above,xshift = -0.3 cm, yshift = -0.4 cm]{Im};
            
            \coordinate (Q1) at (0.2,-0.7);
            \coordinate (Q2) at (0.1,0.3);
            \coordinate (Q3) at (-0.4,-0.1);
            \filldraw[color=blue] (Q1) circle (1.5pt) node[below, xshift=0.2cm]{$p_{2,1}$};
            \filldraw[color=blue] (Q2) circle (1.5pt) node[right]{$p_{2,2}$};
            \filldraw[color=blue] (Q3) circle (1.5pt) node[left, yshift=-0.1cm]{$p_{2,3}$};
            \filldraw[line width=1pt, color=blue, fill=blue, fill opacity=0.3] 
            (Q1) -- (Q2) -- (Q3)-- cycle;
        \end{scope}
        \end{tikzpicture}
    }
    \subfloat[Edges]
    {
        \begin{tikzpicture}[>=stealth]
        \begin{scope}
            \clip (-2.1,-2.1) rectangle (2.1,2.1);
            \draw[line width=0.5pt, ->] (-2.1,0) -- (2.1,0) node (X) [right,xshift = -0.6 cm, yshift = -0.2 cm]{Re};
            \draw[line width=0.5pt, ->] (0,-2.1) -- (0,2.1) node[above,xshift = -0.3 cm, yshift = -0.4 cm]{Im};
            
            \coordinate (P1) at (0.4,-0.7);
            \coordinate (P2) at (0.6,0.7);
            \coordinate (P3) at (-0.4,0.6);
            \coordinate (P4) at (-0.3,-0.4);
            \coordinate (Q1) at (0.2,-0.7);
            \coordinate (Q2) at (0.1,0.3);
            \coordinate (Q3) at (-0.4,-0.1);
            \coordinate (O) at (0,0);

            \draw[line width=1pt, red] (O) -- ($(P2)-(P1)$) node[right]{};
            \draw[line width=1pt, red] (O) -- ($(P3)-(P2)$) node[left]{};
            \draw[line width=1pt, red] (O) -- ($(P4)-(P3)$) node[left]{};
            \draw[line width=1pt, red] (O) -- ($(P1)-(P4)$) node[right]{};
            
            \draw[line width=1pt, blue] (O) -- ($(Q2)-(Q1)$) node[right]{};
            \draw[line width=1pt, blue] (O) -- ($(Q3)-(Q2)$) node[left]{};
            \draw[line width=1pt, blue] (O) -- ($(Q1)-(Q3)$) node[below]{};
        \end{scope}
        \end{tikzpicture}
    }
    \subfloat[Minkowski sum $P_1+P_2$]
    {
        \begin{tikzpicture}[>=stealth]
        \begin{scope}
            \clip (-2.1,-2.1) rectangle (2.1,2.1);
            \draw[line width=0.5pt, ->] (-2.1,0) -- (2.1,0) node (X) [right,xshift = -0.6 cm, yshift = -0.2 cm]{Re};
            \draw[line width=0.5pt, ->] (0,-2.1) -- (0,2.1) node[above,xshift = -0.3 cm, yshift = -0.4 cm]{Im};
            
            \coordinate (P1) at (0.4,-0.7);
            \coordinate (P2) at (0.6,0.7);
            \coordinate (P3) at (-0.4,0.6);
            \coordinate (P4) at (-0.3,-0.4);
            \coordinate (Q1) at (0.2,-0.7);
            \coordinate (Q2) at (0.1,0.3);
            \coordinate (Q3) at (-0.4,-0.1);
            \coordinate (O) at (0,0);

            \coordinate (R1) at ($(P1)+(Q1)$);
            \coordinate (R2) at ($(P2)+(Q1)$);
            \coordinate (R3) at ($(P2)+(Q2)$);
            \coordinate (R4) at ($(P3)+(Q2)$);
            \coordinate (R5) at ($(P3)+(Q3)$);
            \coordinate (R6) at ($(P4)+(Q3)$);
            \coordinate (R7) at ($(P4)+(Q1)$);
            \fill[fill=violet, fill opacity=0.3] 
            (R1) -- (R2) -- (R3) -- (R4) -- (R5) -- (R6) -- (R7) -- cycle;
            \draw[line width=1pt, red] (R1) -- (R2) node[]{};
            \draw[line width=1pt, blue] (R2) -- (R3) node[]{};
            \draw[line width=1pt, red] (R3) -- (R4) node[]{};
            \draw[line width=1pt, blue] (R4) -- (R5) node[]{};
            \draw[line width=1pt, red] (R5) -- (R6) node[]{};
            \draw[line width=1pt, blue] (R6) -- (R7) node[]{};
            \draw[line width=1pt, red] (R7) -- (R1) node[]{};
            \filldraw[color=violet, text=black] (R1) circle (1.5pt) node[below]{$(\textcolor{red}{1},\textcolor{blue}{1})$};
            \filldraw[color=violet, text=black] (R2) circle (1.5pt) node[above right]{$(\textcolor{red}{2},\textcolor{blue}{1})$};
            \filldraw[color=violet, text=black] (R3) circle (1.5pt) node[above]{$(\textcolor{red}{2},\textcolor{blue}{2})$};
            \filldraw[color=violet, text=black] (R4) circle (1.5pt) node[left, xshift=0.1cm, yshift=0.2cm]{$(\textcolor{red}{3},\textcolor{blue}{2})$};
            \filldraw[color=violet, text=black] (R5) circle (1.5pt) node[left]{$(\textcolor{red}{3},\textcolor{blue}{3})$};
            \filldraw[color=violet, text=black] (R6) circle (1.5pt) node[left]{$(\textcolor{red}{4},\textcolor{blue}{3})$};
            \filldraw[color=violet, text=black] (R7) circle (1.5pt) node[left, yshift=-0.1cm]{$(\textcolor{red}{4},\textcolor{blue}{1})$};
        \end{scope}
        \end{tikzpicture}
    }
    \caption{Visualization of Alg. \ref{algo:compute_minkowski_sum} for $N=2$. If two or more edges are parallel, they can be merged and the corresponding vertices removed.}
    \label{fig:algorithm}
\end{figure*}
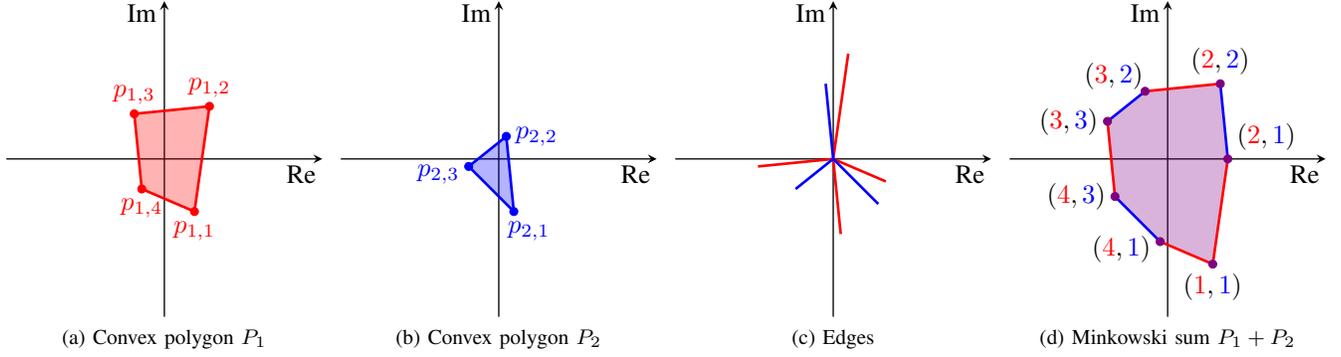

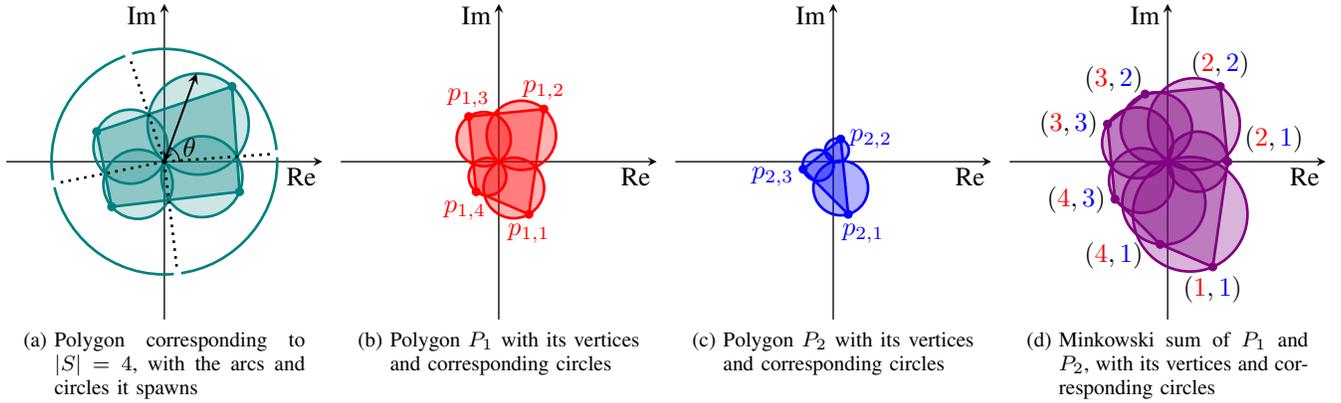
\begin{figure*}
    \centering
    \captionsetup[subfigure]{format=hang, margin=10pt}
    \subfloat[Polygon corresponding to $|S|=4$, with the arcs and circles it spawns]
    {\label{fig:support_function_definition}
        \begin{tikzpicture}[>=stealth]
        \begin{scope}
            \clip (-2.1,-2.1) rectangle (2.1,2.1);
            \draw[line width=0.5pt, ->] (-2.1,0) -- (2.1,0) node (X) [right,xshift = -0.6 cm, yshift = -0.2 cm]{Re};
            \draw[line width=0.5pt, ->] (0,-2.1) -- (0,2.1) node[above,xshift = -0.3 cm, yshift = -0.4 cm]{Im};
            \coordinate (O) at (0,0);
            \coordinate (P1) at (1,-0.4);
            \coordinate (P2) at (0.9,1);
            \coordinate (P3) at (-0.9,0.4);
            \coordinate (P4) at (-0.7,-0.6);
            \filldraw[color=teal] (P1) circle (1.5pt) node[below]{};
            \filldraw[color=teal] (P2) circle (1.5pt) node[above]{};
            \filldraw[color=teal] (P3) circle (1.5pt) node[above]{};
            \filldraw[color=teal] (P4) circle (1.5pt) node[below]{};
            \filldraw[line width=1pt, color=teal, fill=teal, fill opacity=0.3] 
            (P1) -- (P2) -- (P3) -- (P4) -- cycle;

            \coordinate (C) at ($(O)!0.5!(P1)$);
            \pgfpointanchor{C}{center}
            \pgfgetlastxy{\Cx}{\Cy} 
            \pgfmathsetmacro{\Radius}{veclen(\Cx,\Cy)/28.45276}
            \filldraw[line width=1pt, color=teal, fill=teal, fill opacity=0.2] (C) circle (\Radius);

            \coordinate (C) at ($(O)!0.5!(P2)$);
            \pgfpointanchor{C}{center}
            \pgfgetlastxy{\Cx}{\Cy} 
            \pgfmathsetmacro{\Radius}{veclen(\Cx,\Cy)/28.45276}
            \filldraw[line width=1pt, color=teal, fill=teal, fill opacity=0.2] (C) circle (\Radius);

            \coordinate (C) at ($(O)!0.5!(P3)$);
            \pgfpointanchor{C}{center}
            \pgfgetlastxy{\Cx}{\Cy} 
            \pgfmathsetmacro{\Radius}{veclen(\Cx,\Cy)/28.45276}
            \filldraw[line width=1pt, color=teal, fill=teal, fill opacity=0.2] (C) circle (\Radius);

            \coordinate (C) at ($(O)!0.5!(P4)$);
            \pgfpointanchor{C}{center}
            \pgfgetlastxy{\Cx}{\Cy} 
            \pgfmathsetmacro{\Radius}{veclen(\Cx,\Cy)/28.45276}
            \filldraw[line width=1pt, color=teal, fill=teal, fill opacity=0.2] (C) circle (\Radius);

            \coordinate (theta) at (70:1.25);
            \draw[line width=0.75pt, <->] (O)--(theta) node[midway, right, xshift=-0.5mm]{};
            \draw[line width=0.5pt] (0.2,0) arc (0:70:0.2) node[right, xshift=0.5mm]{$\theta$};

            \draw[line width=1pt, color=teal] (O) circle (1.5);
            \coordinate (projection1) at ($(P1)!(O)!(P2)$);
            \coordinate (projection2) at ($(P2)!(O)!(P3)$);
            \coordinate (projection3) at ($(P3)!(O)!(P4)$);
            \coordinate (projection4) at ($(P4)!(O)!(P1)$);
            \coordinate (break1) at ($(O)!1.5cm!(projection1)$);
            \coordinate (break2) at ($(O)!1.5cm!(projection2)$);
            \coordinate (break3) at ($(O)!1.5cm!(projection3)$);
            \coordinate (break4) at ($(O)!1.5cm!(projection4)$);
            \draw[line width=1pt, dotted] (O) -- (break1);
            \draw[line width=1pt, dotted] (O) -- (break2);
            \draw[line width=1pt, dotted] (O) -- (break3);
            \draw[line width=1pt, dotted] (O) -- (break4);
            \fill[white] (break1) circle (0.5mm);
            \fill[white] (break2) circle (0.5mm);
            \fill[white] (break3) circle (0.5mm);
            \fill[white] (break4) circle (0.5mm);
        \end{scope}
        \end{tikzpicture}
    }
    \subfloat[Polygon $P_1$ with its vertices and corresponding circles]
    {
        \begin{tikzpicture}[>=stealth]
        \begin{scope}
            \clip (-2.1,-2.1) rectangle (2.1,2.1);
            \draw[line width=0.5pt, ->] (-2.1,0) -- (2.1,0) node (X) [right,xshift = -0.6 cm, yshift = -0.2 cm]{Re};
            \draw[line width=0.5pt, ->] (0,-2.1) -- (0,2.1) node[above,xshift = -0.3 cm, yshift = -0.4 cm]{Im};
            \coordinate (O) at (0,0);
            \coordinate (P1) at (0.4,-0.7);
            \coordinate (P2) at (0.6,0.7);
            \coordinate (P3) at (-0.4,0.6);
            \coordinate (P4) at (-0.3,-0.4);
            \filldraw[color=red] (P1) circle (1.5pt) node[below]{$p_{1,1}$};
            \filldraw[color=red] (P2) circle (1.5pt) node[above]{$p_{1,2}$};
            \filldraw[color=red] (P3) circle (1.5pt) node[above]{$p_{1,3}$};
            \filldraw[color=red] (P4) circle (1.5pt) node[below,xshift=-0.15cm]{$p_{1,4}$};
            \filldraw[line width=1pt, color=red, fill=red, fill opacity=0.3] 
            (P1) -- (P2) -- (P3) -- (P4) -- cycle;

            \coordinate (C) at ($(O)!0.5!(P1)$);
            \pgfpointanchor{C}{center}
            \pgfgetlastxy{\Cx}{\Cy} 
            \pgfmathsetmacro{\Radius}{veclen(\Cx,\Cy)/28.45276}
            \filldraw[line width=1pt, color=red, fill=red, fill opacity=0.3] (C) circle (\Radius);

            \coordinate (C) at ($(O)!0.5!(P2)$);
            \pgfpointanchor{C}{center}
            \pgfgetlastxy{\Cx}{\Cy} 
            \pgfmathsetmacro{\Radius}{veclen(\Cx,\Cy)/28.45276}
            \filldraw[line width=1pt, color=red, fill=red, fill opacity=0.3] (C) circle (\Radius);

            \coordinate (C) at ($(O)!0.5!(P3)$);
            \pgfpointanchor{C}{center}
            \pgfgetlastxy{\Cx}{\Cy} 
            \pgfmathsetmacro{\Radius}{veclen(\Cx,\Cy)/28.45276}
            \filldraw[line width=1pt, color=red, fill=red, fill opacity=0.3] (C) circle (\Radius);

            \coordinate (C) at ($(O)!0.5!(P4)$);
            \pgfpointanchor{C}{center}
            \pgfgetlastxy{\Cx}{\Cy} 
            \pgfmathsetmacro{\Radius}{veclen(\Cx,\Cy)/28.45276}
            \filldraw[line width=1pt, color=red, fill=red, fill opacity=0.3] (C) circle (\Radius);
            
        \end{scope}
        \end{tikzpicture}
    }
    \subfloat[Polygon $P_2$ with its vertices and corresponding circles]
    {
        \begin{tikzpicture}[>=stealth]
        \begin{scope}
            \clip (-2.1,-2.1) rectangle (2.1,2.1);
            \draw[line width=0.5pt, ->] (-2.1,0) -- (2.1,0) node (X) [right,xshift = -0.6 cm, yshift = -0.2 cm]{Re};
            \draw[line width=0.5pt, ->] (0,-2.1) -- (0,2.1) node[above,xshift = -0.3 cm, yshift = -0.4 cm]{Im};
            
            \coordinate (Q1) at (0.2,-0.7);
            \coordinate (Q2) at (0.1,0.3);
            \coordinate (Q3) at (-0.4,-0.1);
            \filldraw[color=blue] (Q1) circle (1.5pt) node[below, xshift=0.2cm]{$p_{2,1}$};
            \filldraw[color=blue] (Q2) circle (1.5pt) node[right]{$p_{2,2}$};
            \filldraw[color=blue] (Q3) circle (1.5pt) node[left, yshift=-0.1cm]{$p_{2,3}$};
            \filldraw[line width=1pt, color=blue, fill=blue, fill opacity=0.3] 
            (Q1) -- (Q2) -- (Q3)-- cycle;

            \coordinate (C) at ($(O)!0.5!(Q1)$);
            \pgfpointanchor{C}{center}
            \pgfgetlastxy{\Cx}{\Cy} 
            \pgfmathsetmacro{\Radius}{veclen(\Cx,\Cy)/28.45276}
            \filldraw[line width=1pt, color=blue, fill=blue, fill opacity=0.3] (C) circle (\Radius);

            \coordinate (C) at ($(O)!0.5!(Q2)$);
            \pgfpointanchor{C}{center}
            \pgfgetlastxy{\Cx}{\Cy} 
            \pgfmathsetmacro{\Radius}{veclen(\Cx,\Cy)/28.45276}
            \filldraw[line width=1pt, color=blue, fill=blue, fill opacity=0.3] (C) circle (\Radius);

            \coordinate (C) at ($(O)!0.5!(Q3)$);
            \pgfpointanchor{C}{center}
            \pgfgetlastxy{\Cx}{\Cy} 
            \pgfmathsetmacro{\Radius}{veclen(\Cx,\Cy)/28.45276}
            \filldraw[line width=1pt, color=blue, fill=blue, fill opacity=0.3] (C) circle (\Radius);
        \end{scope}
        \end{tikzpicture}
    }
    \subfloat[Minkowski sum of $P_1$ and $P_2$, with its vertices and corresponding circles]
    {\label{fig:support_function_minkowski_sum}
        \begin{tikzpicture}[>=stealth]
        \begin{scope}
            \clip (-2.1,-2.1) rectangle (2.1,2.1);
            \draw[line width=0.5pt, ->] (-2.1,0) -- (2.1,0) node (X) [right,xshift = -0.6 cm, yshift = -0.2 cm]{Re};
            \draw[line width=0.5pt, ->] (0,-2.1) -- (0,2.1) node[above,xshift = -0.3 cm, yshift = -0.4 cm]{Im};
            
            \coordinate (P1) at (0.4,-0.7);
            \coordinate (P2) at (0.6,0.7);
            \coordinate (P3) at (-0.4,0.6);
            \coordinate (P4) at (-0.3,-0.4);
            \coordinate (Q1) at (0.2,-0.7);
            \coordinate (Q2) at (0.1,0.3);
            \coordinate (Q3) at (-0.4,-0.1);
            \coordinate (O) at (0,0);

            \coordinate (R1) at ($(P1)+(Q1)$);
            \coordinate (R2) at ($(P2)+(Q1)$);
            \coordinate (R3) at ($(P2)+(Q2)$);
            \coordinate (R4) at ($(P3)+(Q2)$);
            \coordinate (R5) at ($(P3)+(Q3)$);
            \coordinate (R6) at ($(P4)+(Q3)$);
            \coordinate (R7) at ($(P4)+(Q1)$);
            \filldraw[line width=1pt, color=violet, fill=violet, fill opacity=0.3] 
            (R1) -- (R2) -- (R3) -- (R4) -- (R5) -- (R6) -- (R7) -- cycle;
            \filldraw[color=violet, text=black] (R1) circle (1.5pt) node[below]{$(\textcolor{red}{1},\textcolor{blue}{1})$};
            \filldraw[color=violet, text=black] (R2) circle (1.5pt) node[above right, xshift=0.1cm]{$(\textcolor{red}{2},\textcolor{blue}{1})$};
            \filldraw[color=violet, text=black] (R3) circle (1.5pt) node[above]{$(\textcolor{red}{2},\textcolor{blue}{2})$};
            \filldraw[color=violet, text=black] (R4) circle (1.5pt) node[left, xshift=0.1cm, yshift=0.2cm]{$(\textcolor{red}{3},\textcolor{blue}{2})$};
            \filldraw[color=violet, text=black] (R5) circle (1.5pt) node[left]{$(\textcolor{red}{3},\textcolor{blue}{3})$};
            \filldraw[color=violet, text=black] (R6) circle (1.5pt) node[left]{$(\textcolor{red}{4},\textcolor{blue}{3})$};
            \filldraw[color=violet, text=black] (R7) circle (1.5pt) node[left, xshift =-0.1cm, yshift=-0.15
            cm]{$(\textcolor{red}{4},\textcolor{blue}{1})$};

            \coordinate (C) at ($(O)!0.5!(R1)$);
            \pgfpointanchor{C}{center}
            \pgfgetlastxy{\Cx}{\Cy} 
            \pgfmathsetmacro{\Radius}{veclen(\Cx,\Cy)/28.45276}
            \filldraw[line width=1pt, color=violet, fill=violet, fill opacity=0.3] (C) circle (\Radius);

            \coordinate (C) at ($(O)!0.5!(R2)$);
            \pgfpointanchor{C}{center}
            \pgfgetlastxy{\Cx}{\Cy} 
            \pgfmathsetmacro{\Radius}{veclen(\Cx,\Cy)/28.45276}
            \filldraw[line width=1pt, color=violet, fill=violet, fill opacity=0.3] (C) circle (\Radius);

            \coordinate (C) at ($(O)!0.5!(R3)$);
            \pgfpointanchor{C}{center}
            \pgfgetlastxy{\Cx}{\Cy} 
            \pgfmathsetmacro{\Radius}{veclen(\Cx,\Cy)/28.45276}
            \filldraw[line width=1pt, color=violet, fill=violet, fill opacity=0.3] (C) circle (\Radius);

            \coordinate (C) at ($(O)!0.5!(R4)$);
            \pgfpointanchor{C}{center}
            \pgfgetlastxy{\Cx}{\Cy} 
            \pgfmathsetmacro{\Radius}{veclen(\Cx,\Cy)/28.45276}
            \filldraw[line width=1pt, color=violet, fill=violet, fill opacity=0.3] (C) circle (\Radius);

            \coordinate (C) at ($(O)!0.5!(R5)$);
            \pgfpointanchor{C}{center}
            \pgfgetlastxy{\Cx}{\Cy} 
            \pgfmathsetmacro{\Radius}{veclen(\Cx,\Cy)/28.45276}
            \filldraw[line width=1pt, color=violet, fill=violet, fill opacity=0.3] (C) circle (\Radius);

            \coordinate (C) at ($(O)!0.5!(R6)$);
            \pgfpointanchor{C}{center}
            \pgfgetlastxy{\Cx}{\Cy} 
            \pgfmathsetmacro{\Radius}{veclen(\Cx,\Cy)/28.45276}
            \filldraw[line width=1pt, color=violet, fill=violet, fill opacity=0.3] (C) circle (\Radius);

            \coordinate (C) at ($(O)!0.5!(R7)$);
            \pgfpointanchor{C}{center}
            \pgfgetlastxy{\Cx}{\Cy} 
            \pgfmathsetmacro{\Radius}{veclen(\Cx,\Cy)/28.45276}
            \filldraw[line width=1pt, color=violet, fill=violet, fill opacity=0.3] (C) circle (\Radius);
        \end{scope}
        \end{tikzpicture}
    }
    \caption{For $N=2$ with set cardinalities $4$ and $3$, visualization of the circles spawn by the respective polygons, and their Minkowski sum.
    }
    \label{fig:support_function}
\end{figure*}


\section{Enumeration of Vertices}

As per Thm. \ref{theorem:evaluation_at_vertices}, it is enough to enumerate the vertices of 
\begin{align}
\label{fred}
    h_1\,\conv\,\Theta_1 + \cdots + h_N\, \conv\,\Theta_N
\end{align}
and evaluate their moduli. Although the cardinality of the set
$
     h_1\Theta_1 + \cdots + h_N \Theta_N
$
can be as large as $\prod_n |\Theta_n|$, the number of vertices of \eqref{fred}
is at most $\sum_n |\Theta_n|$ 
\cite[Ch. 13.3]{de2008computational}.
The Minkowski sum of $N$ convex polygons is another convex polygon whose boundary can be computed by merging the edges of the polygons after sorting them by orientation. The procedure is detailed for arbitrary $N$ in Alg. \ref{algo:compute_minkowski_sum} and illustrated for $N=2$ in Fig.~\ref{fig:algorithm}. 


As the computation of edges can be done in linear time, the dominant step is the sorting of $K = \sum_n |\Theta_n|$ edges, which is merely $O(\log K!) = O(K \log K)$  \cite{martin1971sorting}.
In the special case
\begin{align}
    \Theta_1 = \cdots = \Theta_N = \left\{  e^{j2\pi\frac{0}{M}},e^{j2\pi\frac{1}{M}}, \ldots, e^{j2\pi\frac{M-1}{M}} \right\} ,
\end{align}
it suffices to compute for each $n$ one edge whose argument lies between $0$ and $\frac{2\pi}{M}$, thanks to the rotational symmetry \cite{pekcan2024achieving}. This means that only $N$ edges need to be sorted.

\section{Connection to Prior-Art Approaches}

Most prior derivations are based on the formula
\begin{align}
    |z| = \max_{\theta\in[0,2\pi]} \Re{e^{-j\theta} z},
\end{align}
which turns the nonlinear function $z\mapsto |z|$ into a maximum of linear functions.
This linearity enables decoupling the objective in \eqref{original_problem} into
\begin{align}
    \max_{w_1,\ldots,w_N} \bigg|\sum_n w_n h_n\bigg| &= \max_{w_1,\ldots,w_N} \max_{\theta\in[0,2\pi]} \Re{\!e^{-j\theta} \sum_n w_n h_n\!} \nonumber \\
    &= \max_{\theta\in[0,2\pi]} \sum_n \max_{w_n \in \Theta_n}  \Re{e^{-j\theta} w_n h_n} \nonumber \\
    &= \max_{\theta\in[0,2\pi]} \sum_n \max_{z \in h_n\Theta_n}\!\!  \Re{e^{-j\theta} z}. \label{linearity}
\end{align}
%

For a finite set $S=\{z_1, \ldots, z_{|S|}\}$ forming a convex polygon, as illustrated in Fig. \ref{fig:support_function_definition},
\begin{align}
    \max\limits_{z \in S}\Re{e^{-j\theta} z} = \begin{cases}
        \Re{e^{-j\theta} z_1} &\text{if $e^{j\theta}$ is in arc $1$}\\
        \hspace*{2.5em}\vdots\\
        \Re{e^{-j\theta} z_{|S|}} &\text{if $e^{j\theta}$ is in arc $|S|$}
    \end{cases}, \nonumber
\end{align}
where the complex unit circle was partitioned into $|S|$ arcs.
The mapping 
\begin{align}
    e^{j\theta}\mapsto \max_{z \in S}  \Re{e^{-j\theta} z} \label{support_function}
\end{align}
can be visualized as $|S|$ circles, one for $\theta$ sweeping each of the arcs. 
Indeed, for fixed $z$, the complex number
\begin{align}
    e^{j\theta}\Re{e^{-j\theta} z}
\end{align}
is on the circle having the origin and $z$ as endpoints of a diameter; this follows from Thale's theorem in geometry,
\begin{align}
    \bigg|e^{j\theta}\Re{e^{-j\theta} z} - \frac{z}{2}\bigg|^2 &= \frac{1}{4}\big|e^{j\theta}(e^{-j\theta} z + e^{j\theta} \overline{z}) - z\big|^2 \\
    & = \frac{|z|^2}{4}. 
\end{align}
More generally, 
\begin{align}
    e^{j\theta}\mapsto \sum_n \max_{z \in h_n\Theta_n}\!\!  \Re{e^{-j\theta} z}, \label{sum_of_support_functions}
\end{align}
can be visualized as at most $K = \sum_n|\Theta_n|$ circles, obtained efficiently by sorting the breakpoints. (This is exemplified for $N=2$ in Fig. \ref{fig:support_function_minkowski_sum}, where $(\ell_1,\ell_2)$ is the circle spawn by $e^{j\theta}$ in the $\ell_1$th arc of $P_1$ and the $\ell_2$th arc of $P_2$.)
It follows that
\begin{align}
        \sum_n \max_{z \in h_n\Theta_n}\!\!  \Re{e^{-j\theta} z} \leq \text{maximum of $K$ diameters} .
    \end{align}
As every circle has the origin as well as some point in $h_1\Theta_1 + \cdots + h_N\Theta_N$ as endpoints of a diameter, there exist $w_1, \ldots, w_N$ such that
    \begin{align}
        \text{maximum of $K$ diameters} \leq \bigg|\sum_n w_n h_n\bigg|.
    \end{align}
Altogether,
\begin{align}
    \max_{\theta\in[0,2\pi]} \sum_n \max_{z \in h_n\Theta_n}\!\!  \Re{e^{-j\theta} z} & \leq \text{maximum of $K$ diameters} \nonumber \\
    & \leq \max_{w_1,\ldots,w_N} \bigg|\sum_n w_n h_n\bigg|,
\end{align}
which holds as equality by virtue of \eqref{linearity}.
Existing proofs count the number of arcs in Fig. \ref{fig:support_function}.
Our approach instead counts the number of vertices, which is equivalent---this equivalence is termed \textit{duality} for $\Theta_n = \{0,1\}$ in \cite{allemand2001polynomial}---but more direct.


\appendix

Applying Lemma \ref{lemma:maximum_principle} (presented below), the optimization problem can be recast as
\begin{align}
    \max_{z} & \quad |z| \\
    \text{s.t.} & \quad z\in h_1\,\conv\,\Theta_1 + \cdots + h_N\, \conv\,\Theta_N, \nonumber
\end{align}
where the property of the Minkowski sum,
\begin{align}
    &\conv(h_1\Theta_1 + \cdots + h_N\Theta_N) \nonumber \\
    & \qquad\qquad=\conv(h_1\Theta_1) + \cdots + \conv(h_N\Theta_N)\\
    & \qquad\qquad =h_1\,\conv\,\Theta_1 + \cdots + h_N\, \conv\,\Theta_N, \label{minkowski}
\end{align}
was used.
The constraint set is the Minkowski sum of convex polygons, which is a convex polygon. Applying Lemma \ref{lemma:maximum_principle} again, one can surmise that it suffices to evaluate $|z|$ at the vertices, as the convex polygon is a convex hull of its vertices.

A related idea has been propounded to solve zero-one quadratic optimizations \cite{allemand2001polynomial}.

\begin{lemma}
\label{lemma:maximum_principle}
For any closed and bounded set $S\subset \bbC$,
\begin{align}
    \max_{z\in S} |z| = \max_{z\in \conv S} |z|.
\end{align}
\end{lemma}
\begin{proof}
From $S \subset \conv S$, it is trivial that
\begin{align}
    \max_{z\in S} |z| \leq \max_{z\in \conv S} |z|,
\end{align}
and it is enough to show the inequality in the reverse direction. For $z\in\conv \,S$, by definition of convex hull, $z$ can be written as a convex combination of the elements of $S$. That is, there exist $\lambda_k\geq 0$ and $z_k\in S$ such that
\begin{align}
    z = \sum_k \lambda_k z_k \qquad \sum_k \lambda_k = 1.
\end{align}
From the triangle inequality 
\begin{align}
    |z| = \bigg|\sum_k \lambda_k z_k\bigg|\leq \sum_k \lambda_k |z_k|\leq \max_{z\in S} |z| .
\end{align}
\end{proof}

\bibliographystyle{IEEEtran}
\bibliography{ref}

@ARTICLE{ren2023ieee,
  author={Ren, Shuyi and Shen, Kaiming and Li, Xin and Chen, Xin and Luo, Zhi-Quan},
  journal={IEEE Wireless Commun. Lett.}, 
  title={A Linear Time Algorithm for the Optimal Discrete {IRS} Beamforming}, 
  year={2023},
  volume={12},
  number={3},
  pages={496-500}
}

@article{vardakis2023intelligently,
  title={Intelligently Wireless Batteryless {RF}-Powered Reconfigurable Surface: Theory, Implementation \& Limitations},
  author={Vardakis, Iosif and Kotridis, Georgios and Peppas, Spyridon and Skyvalakis, Konstantinos and Vougioukas, Georgios and Bletsas, Aggelos},
  volume={22},
  number={6},
  pages={3942-3954},
  journal={IEEE Trans. Wireless Commun.},
  year={2023},
  publisher={IEEE}
}

@ARTICLE{pekcan2024achieving,
  author={Pekcan, Dogan Kutay and Ayanoglu, Ender},
  journal={IEEE Open J. Commun. Soc.}, 
  title={Achieving Optimum Received Power for Discrete-Phase {RISs} With Elementwise Updates in the Least Number of Steps}, 
  year={2024},
  volume={5},
  number={},
  pages={2706-2722}
}

@ARTICLE{zhang2022configuring,
  author={Zhang, Yaowen and Shen, Kaiming and Ren, Shuyi and Li, Xin and Chen, Xin and Luo, Zhi-Quan},
  journal={IEEE J. Sel. Topics Signal Process.}, 
  title={Configuring Intelligent Reflecting Surface With Performance Guarantees: Optimal Beamforming}, 
  year={2022},
  volume={16},
  number={5},
  pages={967-979}
}

@ARTICLE{kutay2024received,
  author={Kutay Pekcan, Dogan and Liao, Hongyi and Ayanoglu, Ender},
  journal={IEEE Open J. Commun. Soc.}, 
  title={Received Power Maximization Using Nonuniform Discrete Phase Shifts for {RISs} With a Limited Phase Range}, 
  year={2024},
  volume={5},
  number={},
  pages={7447-7466}
}

@ARTICLE{hashemi2024optimal,
  author={Hashemi, Seyedkhashayar and Jiang, Hai and Ardakani, Masoud},
  journal={IEEE Trans. Commun.}, 
  title={Optimal Configuration of Reconfigurable Intelligent Surfaces With Arbitrary Discrete Phase Shifts}, 
  year={2024},
  volume={72},
  number={12},
  pages={8047-8060}
}

@ARTICLE{sanjay2024optimum,
  author={Sanjay Narayanan, Sai and Khankhoje, Uday K. and Krishna Ganti, Radha},
  journal={IEEE Trans. Antennas Propag.}, 
  title={Optimum Beamforming and Grating-Lobe Mitigation for Intelligent Reflecting Surfaces}, 
  year={2024},
  volume={72},
  number={11},
  pages={8540-8553}
}

@ARTICLE{alevizos2016log,
  author={Alevizos, Panos N. and Fountzoulas, Yannis and Karystinos, George N. and Bletsas, Aggelos},
  journal={IEEE Trans. Commun.}, 
  title={Log-Linear-Complexity {GLRT}-Optimal Noncoherent Sequence Detection for Orthogonal and {RFID}-Oriented Modulations}, 
  year={2016},
  volume={64},
  number={4},
  pages={1600-1612}
}

@INPROCEEDINGS{sanchez2021optimal,
  author={Sanchez, Juan and Bengtsson, Erik and Rusek, Fredrik and Flordelis, Jose and Zhao, Kun and Tufvesson, Fredrik},
  booktitle={IEEE Global Commun. Conf.}, 
  title={Optimal, Low-Complexity Beamforming for Discrete Phase Reconfigurable Intelligent Surfaces}, 
  year={2021},
  volume={},
  number={},
  pages={01-06}
}

@ARTICLE{deng2019mmwave,
  author={Deng, Junquan and Tirkkonen, Olav and Studer, Christoph},
  journal={IEEE Trans. Veh. Technol.}, 
  title={{MmWave} Multiuser {MIMO} Precoding With Fixed Subarrays and Quantized Phase Shifters}, 
  year={2019},
  volume={68},
  number={11},
  pages={11132-11145}
}

@ARTICLE{motedayenaval2003polynomial,
  author={Motedayen-Aval, I. and Anastasopoulos, A.},
  journal={IEEE Trans. Commun.}, 
  title={Polynomial-complexity noncoherent symbol-by-symbol detection with application to adaptive iterative decoding of turbo-like codes}, 
  year={2003},
  volume={51},
  number={2},
  pages={197-207}
}

@ARTICLE{mackenthun1994fast,
  author={Mackenthun, K.M.},
  journal={IEEE Trans. Commun.}, 
  title={A fast algorithm for multiple-symbol differential detection of {MPSK}}, 
  year={1994},
  volume={42},
  number={234},
  pages={1471-1474}
}

@ARTICLE{sweldens2001fast,
  author={Sweldens, W.},
  journal={IEEE Commun. Lett.}, 
  title={Fast block noncoherent decoding}, 
  year={2001},
  volume={5},
  number={4},
  pages={132-134}
}

@book{de2008computational,
  title={Computational geometry: algorithms and applications},
  author={De Berg, Mark and Cheong, Otfried and Van Kreveld, Marc and Overmars, Mark},
  year={2008},
  edition={third},
  publisher={Springer}
}

@article{allemand2001polynomial,
  title={A polynomial case of unconstrained zero-one quadratic optimization},
  author={Allemand, Kim and Fukuda, Komei and Liebling, Thomas M and Steiner, Erich},
  journal={Math. Program.},
  volume={91},
  number={1},
  pages={49--52},
  year={2001},
  publisher={Springer}
}

@article{martin1971sorting,
  title={Sorting},
  author={Martin, William A.},
  journal={ACM Comput. Surv.},
  volume={3},
  number={4},
  pages={147--174},
  year={1971},
  publisher={ACM New York, NY, USA}
}

\end{document}